\documentclass[11pt]{article}
\usepackage{amsthm,amsfonts,amsmath,times}
\usepackage{enumerate, graphicx}
\usepackage{algo, extras, xspace}
\usepackage{boxedminipage, wrapfig, picins}
\usepackage{url}

\def\Comment#1{\textsl{$\langle\!\langle$#1\/$\rangle\!\rangle$}}

\advance \topmargin by -1.0in
\advance \oddsidemargin by -0.8in
\newcommand{\myparskip}{3pt}
\parskip \myparskip

\newtheorem{lemma}{Lemma}[section]
\newtheorem{theorem}[lemma]{Theorem}

\newtheorem{corollary}[lemma]{Corollary}
\newtheorem{prop}[lemma]{Proposition}

\newcommand{\opt}{\textsc{OPT}}
\newcommand{\optcr}{\textsc{OPT}_{\textsc{frac}}}
\newcommand{\costcr}{\textsc{COST}_{\textsc{frac}}}
\newcommand{\partcr}{\textsc{PART}_{\textsc{frac}}}

\newcommand{\costint}{\textsc{COST}_{\textsc{int}}}
\newcommand{\partint}{\textsc{PART}_{\textsc{int}}}
\newcommand{\costballs}{\textsc{COST}_{\textsc{balls}}}
\newcommand{\partballs}{\textsc{PART}_{\textsc{balls}}}
\newcommand{\etal}{{\em et al.}~}
\DeclareMathOperator*{\Ex}{\mathbb{E}}

\renewenvironment{proof}{\vspace{-0.1in}\noindent{\bf Proof:}}%
        {\hspace*{\fill}$\Box$\par}
\newenvironment{proofof}[1]{\smallskip\noindent{\bf Proof of #1:}}%
        {\hspace*{\fill}$\Box$\par}
\newenvironment{proofsketch}{\vspace{-0.1in}\noindent{\bf Proof Sketch:}}%
        {\hspace*{\fill}$\Box$\par}

\def\eps{\varepsilon}

\def\script#1{\mathcal{#1}}
\def\b1{{\bf 1}}
\def\bx{{\bf x}}
\def\bd{{\bf d}}
\def\sep{\;|\;}

\def\CR{\textsc{LE-Rel}\xspace}
\def\MCSAfull{Minimum Submodular-Cost Allocation\xspace}
\def\MCSA{\textsc{MSCA}\xspace}
\def\monMCSA{\textsc{Monotone-MSCA}\xspace}
\def\monGreedy{\textsc{Monotone-MSCA-Greedy}\xspace}
\def\SubMPfull{\textsc{Submodular Multiway Partition}\xspace}
\def\SubMP{\textsc{Sub-MP}\xspace}
\def\SubMPRel{\textsc{SubMP-Rel}\xspace}

\def\SymSubMP{\textsc{Sym-Sub-MP}\xspace}
\def\MMCSA{\textsc{Monotone MSCA}\xspace}
\def\AHMCfull{\textsc{Hypergraph Multiway Cut}\xspace}
\def\AHMC{\textsc{Hypergraph-MC}\xspace}
\def\HMPfull{\textsc{Hypergraph Multiway Partition}\xspace}
\def\HMP{\textsc{Hypergraph-MP}\xspace}

\def\SubMLfull{\textsc{Submodular Cost Labeling}\xspace}
\def\SubML{\textsc{Sub-Label}\xspace}

\def\CKR{\textsc{CKR-Rounding}\xspace}
\def\HR{\textsc{Half-Rounding}\xspace}
\def\KT{\textsc{KT-Rounding}\xspace}
\def\SymMPR{\textsc{SymSubMP-Rounding}\xspace}
\def\SymMLR{\textsc{SymSubLabel-Rounding}\xspace}
\def\nodeMC{\textsc{Node-wt-MC}\xspace}
\def\MC{\textsc{Graph-MC}\xspace}

\begin{document}
\title{Submodular Cost Allocation Problem and
Applications\footnote{An extended abstract of this paper 
will appear in {\em Proc.\ of ICALP}, July 2011.}}
\author{
Chandra Chekuri\thanks{Dept. of Computer Science, University of Illinois, 
Urbana, IL 61801. Supported in part by NSF grants CCF-0728782 and
CCF-1016684. {\tt chekuri@cs.illinois.edu}}
\and
Alina Ene\thanks{Dept. of Computer Science, University of Illinois, Urbana,
    IL 61801. Supported in part by NSF grants CCF-0728782 and CCF-1016684.
{\tt  ene1@illinois.edu}}
}
\date{\today}

\maketitle

\thispagestyle{empty}

\begin{abstract}
We study the \MCSAfull problem (\MCSA). In this problem we are given a
finite ground set $V$ and $k$ non-negative submodular set functions
$f_1,\ldots,f_k$ on $V$. The objective is to partition $V$
into $k$ (possibly empty) sets $A_1, \cdots, A_k$ such
that the sum $\sum_{i = 1}^k f_i(A_i)$ is minimized. Several well-studied
problems such as the non-metric facility location problem, multiway-cut in
graphs and hypergraphs, and uniform metric labeling and its generalizations
can be shown to be special cases of \MCSA.
In this paper we consider a convex-programming relaxation obtained via
the Lov\'{a}sz-extension for submodular functions. This allows us to
understand several previous relaxations and rounding procedures in a
unified fashion and also develop new formulations and approximation
algorithms for several problems. In particular, we give a
$(1.5 - 1/k)$-approximation for the hypergraph multiway partition problem. We
also give a $\min\{2(1-1/k), H_\Delta\}$-approximation for the
hypergraph multiway cut problem when $\Delta$ is the maximum hyperedge
size. Both problems generalize the multiway cut problem in graphs
and the hypergraph cut problem is approximation equivalent to the
node-weighted multiway cut problem in graphs.
\end{abstract}

\newpage
\setcounter{page}{1}
\section{Introduction}
We consider the following allocation problem with submodular costs.

\noindent
\textbf{\MCSAfull}~(\MCSA).  Let $V$ be a finite ground set and let
$f_1, \cdots, f_k$ be $k$ non-negative submodular set functions on
$V$. That is, for $1 \le i \le k$, $f_i:2^V \rightarrow \mathbb{R}_+$
and $f_i(A) + f_i(B) \ge f_i(A \cup B) + f_i(A \cap B)$ for all $A, B
\subseteq V$. In the \textsc{\MCSA}~problem the goal is to partition
the ground set $V$ into $k$ (possibly empty) sets $A_1, \cdots, A_k$ such
that the sum $\sum_{i = 1}^k f_i(A_i)$ is minimized.

We observe that the problem is interesting only if the $f_i$'s are
different for otherwise allocating all of $V$ to $f_1$ is trivially an
optimal solution. We assume that the functions $f_i$ are given as
value oracles, although in specific applications they may be available
as explicit poly-time computable functions of some auxiliary input.
The special case of this problem in which all of the functions are
monotone ($f(A) \le f(B)$ if $A \subseteq B$) has been previously
considered by Svitkina and Tardos \cite{SvitkinaT06}. In this paper,
we consider the problem with both monotone and non-monotone functions.
We show that several well-studied problems such as non-metric facility
location, multiway cut problems in graphs and hypergraphs, uniform
metric labeling and its generalization to hub location among others
can be cast as special cases of \MCSA. In particular, we investigate
the integrality gap of a simple and natural convex-programming
relaxation for \MCSA that is obtained via the use of the Lov\'asz
extension of a submodular function.

\smallskip
\noindent
\textbf{Lov\'asz extension and a convex program for \MCSA:} Let $V$
be a finite ground set of cardinality $n$.
Each real-valued set function on $V$ corresponds to a function
$f:\{0,1\}^n \rightarrow \mathbb{R}$ on the vertices of the
$n$-dimensional hypercube. The Lov\'asz extension of $f$ to the
continuous domain $[0,1]^n$ denoted by $\hat{f}$ is defined
as\footnote{The definition is not the standard one but is equivalent
to it; see \cite{Vondrak09} or Appendix~\ref{app:lovasz}. This
definition is convenient to us in describing and understanding
rounding procedures.}
	$$\hat{f}(\bx) = \Ex_{\theta \in
	[0,1]}\left[f(\bx^{\theta})\right] =
	\int_0^1 f(\bx^{\theta}) d\theta$$
where $\bx^{\theta} \in \{0, 1\}^n$ for a given vector $\bx \in
[0,1]^n$ is defined as: $x_i^{\theta} = 1$ if $x_i \geq \theta$ and
$0$ otherwise.

Lov\'asz showed that $\hat{f}$ is convex if and only if $f$ is a
submodular set function \cite{Lovasz83}. Moreover, it is easy to see
that, given $\bx$, the value $\hat{f}(\bx)$ can be computed in
polynomial time by using a value oracle for $f$.  Via this extension,
we obtain a straightforward relaxation for \MCSA with a convex
objective function and linear constraints.  Let $v_1, \cdots, v_n$
denote the elements of $V$. The relaxation has variables $x(v,i)$ for
$v \in V$ and $1 \le i \le k$ with the interpretation that $x(v,i)$ is
$1$ if $v$ is assigned to $A_i$ and $0$ otherwise.  Let $\bx_i =
(x(v_1, i), \cdots, x(v_n, i))$. The relaxation is given below.

\begin{center}
\begin{boxedminipage}{0.5\linewidth}
\vspace{-0.2in}
\begin{align*}
& \textbf{\CR} &\\
\min \qquad & \sum_{i = 1}^k \hat{f}_i(\bx_i) &\\
& \sum_{i = 1}^k x(v, i) = 1 & \forall v\\
& x(v,i) \geq 0  & \forall v, i
\end{align*}
\end{boxedminipage}
\end{center}

\smallskip
\noindent
Throughout, we use $\opt$ and $\optcr$ to denote the value of an
optimal integral and an optimal fractional solution to
\CR (respectively).

We remark that \CR can be solved in time that is polynomial in $n$ and
$\log\left(\max_{S \subseteq V} f(S)\right)$ via the ellipsoid method;
we give some of the details in Appendix~\ref{app:lovasz}.
Moreover, for some problems of interest the above convex program can
be rewritten into an equivalent linear program. We now describe
several problems that can be cast as special cases of \MCSA, and also
how some previously considered linear-programming relaxations can be
seen as being equivalent to the convex program above.

\subsection{Problems related to \MCSA}
\textbf{Monotone \MCSA (\monMCSA) and Facility Location}: In
facility location, we have a set of facilities $\mathcal{F}$ and a set
of clients or demands $\mathcal{D}$. There is a non-negative cost
$c_{ij}$ to connect facility $i$ to client $j$ (we do not necessarily
assume that these costs form a metric). Opening facility $i \in
\mathcal{F}$ costs $f_i$. The goal is to open a subset of the
facilities and assign each client to an open facility so as to
minimize the sum of the facility opening cost and the connection
costs. Svitkina and Tardos \cite{SvitkinaT06} considered the setting
where the cost of opening a facility $i$ is a monotone submodular
function $g_i$ of the clients assigned to it, and gave an $(1+\ln
|\mathcal{D}|)$-approximation, and matching hardness via a reduction
from set cover. We note that this problem is equivalent to \MCSA
when all the $f_i$ are monotone submodular functions, which we refer to
as \monMCSA. In \cite{SvitkinaT06} a greedy algorithm via submodular
function minimization is used to derive the approximation. Here we
prove that the integrality gap of \CR is $(1+\ln |\mathcal{D}|)$, and
describe how certain rounding algorithms achieve this bound. These
algorithms are useful when considering functions that are not
necessarily monotone.

\smallskip
\noindent 
\textbf{Submodular Multiway Partition (\SubMP)}: We define an abstract
problem and then specialize to known problems. Let $f:2^V \rightarrow
\mathbb{R}_+$ be a submodular set function over $V$ and let $S =
\{s_1,s_2,\ldots,s_k\}$ be $k$ terminals in $V$. The submodular
multiway partition problem is to find a partition of $V$ into
$A_1,\ldots,A_k$ such that $s_i \in A_i$ and $\sum_{i=1}^k f(A_i)$ is
minimized. This has been previously considered by Zhao, Nagamochi and
Ibaraki \cite{ZhaoNI05}. This can be seen as a special case of \MCSA
as follows. Define the ground set to be $V'=V\setminus S$ and, for $1
\le i \le k$, $f_i:2^{V'}\rightarrow \mathbb{R}_+$ is the function
defined as $f_i(S) = f(S \cup \{s_i\})$.  If in addition $f$ is
symmetric ($f(A) = f(V-A)$ for all $A$) we call this the symmetric
\SubMP problem (\SymSubMP). Note that although the problem is based
on a single function $f$, $k$ different submodular functions (induced
by the terminals) are needed to reduce it to \MCSA. We now discuss
some important special cases of this problem.

\noindent
{\em Multiway Cut in Graphs} (\MC): The input is an edge-weighted
undirected graph $G=(V,E)$ and $k$ terminal vertices $S = \{s_1,
\ldots, s_k\}$; the goal is to remove a minimum-weight set of edges to
disconnect the terminals. This can be seen as a special case of the
symmetric submodular multiway partition problem by simply choosing $f$
to be the cut-capacity function of $G$ scaled down by a factor of $2$.
That is, $f(A) = \frac12 \sum_{e \in \delta(A)} w(e)$ where $w(e)$ is
the weight of edge $e$. We observe that \CR for this problem is
equivalent to the well-known geometric LP relaxation of Calinescu,
Karloff and Rabani \cite{CalinescuKR98}, which led to significant
improvements ($1.5-1/k$ in \cite{CalinescuKR98} and $1.3438$ in
\cite{KargerKSTY99}) over the $2(1-1/k)$-approximation obtained via
the isolating-cut heuristic \cite{DahlhausJPSY92}.

\noindent
{\em Multiway Cut and Partition in Hyper-Graphs}: Given an
edge-weighted hypergraph $\script{G}=(V,\script{E})$ and terminal set
$S \subset V$, the \AHMCfull problem (\AHMC)
(see \cite{OkumotoFN10,Xiao10,Fukunaga10}) asks for the minimum weight
subset of hyperedges whose removal disconnects the terminals. This can
be seen as a special case of \SubMP \cite{OkumotoFN10}; this reduction
requires some care and the underlying submodular function is {\em
asymmetric}. A related problem is the \HMPfull problem (\HMP)
introduced by Lawler \cite{Lawler73} where the cost for hyperedge $e$
is proportional to the number of non-trivial pieces it is partitioned
into. This can be seen as a special case of the \SymSubMP
with $f$ being the hypergraph cut capacity function. We note
that \MC is a special case of both \AHMC and \HMP.

\noindent
{\em Node-weighted Multiway Cut in Graphs} (\nodeMC): In this problem
\cite{GargVY04} the graph has weights on nodes instead of edges and
the goal is to find a minimum weight subset of nodes whose removal
disconnects a given set of terminals. It is not difficult to show that
\AHMC and \nodeMC are approximation equivalent \cite{OkumotoFN10}.

Zhao \etal \cite{ZhaoNI05} consider generalizations of the above
problems where some set of terminals $S \subseteq V$ and $k$ are
specified and the goal is to partition $V$ into $k$ sets such that
each set contains at least one terminal and the total cost of the
partition is minimized. We do not discuss these further since they are
not directly related to \MCSA, although one can reduce them to \MCSA
if $k$ is a fixed constant.

\vspace{2mm}
\noindent
\textbf{Uniform Metric Labeling and \SubMLfull (\SubML):} The metric
labeling problem was introduced by Kleinberg and Tardos
\cite{KleinbergT02} as a general classification problem. We are given
an undirected edge-weighted graph $G=(V,E)$ and $k$ labels and the
goal is to assign a label to each vertex to minimize the labeling cost
and the edge-cut cost.  Assigning label $i$ to $v$ incurs a cost
$c_i(v)$ and if an edge $uv$ of weight $w(uv)$ has $u$ labeled with
$i$ and $v$ labeled with $j$ then the edge-cut cost incurred is $w(uv)
\cdot d(ij)$.  The uniform metric labeling problem is obtained when
$d(ij) = 1$ for all $i \neq j$. We consider the following
generalization that we call the \SubMLfull (\SubML) problem which is a
special case of \MCSA.  The $k$ labels correspond to the $k$ functions
$f_1,\ldots,f_k$.  We define $f_i$ as the sum of two functions, a
monotone function $g_i$ that models the label assignment cost, and a
non-monotone function $h$ that models the cut-cost. The goal then is
to partition $V$ into $A_1,\ldots,A_k$ to minimize $\sum_{i=1}^k
(g_i(A_i) + h(A_i))$. Note that uniform metric labeling is the special
case when $g_i$ are modular and $h$ is the graph cut function, which
is symmetric. We are motivated to consider this generalization by
problems that have been considered previously, such as metric labeling
on hypergraphs, hub location problem \cite{GeYZ07}, and the extension
of metric labeling to handle label opening costs \cite{DelongOIB10}.

\subsection{Overview of Results and Techniques}
In this paper we examine the complexity of \MCSA primarily through the
``integrality gap'' of the convex relaxation \CR which can be
optimized in polynomial time.  All the problems we consider are
NP-hard and our focus is on polynomial time approximation algorithms.

A significant portion of our contribution is to highlight the
naturalness of \MCSA and the Lov\'asz-extension based relaxation \CR
by showing connections to previously studied problems, linear
programming relaxations, and rounding strategies. Viewing these
problems in the more abstract setting of submodularity gives insights
into prior algorithms. In the process, we obtain new and interesting
results. Although one would like to obtain a single unifying algorithm
that achieves a good approximation for \MCSA, it turns out that \CR
has a large integrality gap and we believe that \MCSA is hard to
approximate to a polynomial factor. However, it is fruitful to examine
special cases of \MCSA that admit good approximations via \CR. We
describe several applications below by summarizing our results; all of
them are based on \CR.
\begin{itemize}
\item The integrality gap of \CR for \monMCSA is $\Theta(\log n)$.
\item There is a $(1.5-1/k)$-approximation for \HMP.
\item There is a $\min\{2(1-1/k), H_\Delta\}$-approximation for \AHMC,
  where $\Delta$ is the maximum hyperedge size and $H_i$ is the
  $i$-th harmonic number. For $\Delta=2$ this gives a
  $1.5$-approximation and for $\Delta=3$ this gives a
  $1.833$-approximation. 
\item \CR for \AHMC gives a new mathematical programming relaxation for \nodeMC
  and a new $2$-approximation. Moreover, if all non-terminal
  nodes have degree at most $3$ we obtain a $1.833$-approximation
  improving upon the $2(1-1/k)$ known via the distance-based
  relaxation \cite{GargVY04}.
\item The integrality gap of \CR for \SymSubMP is at most $2-2/k$;
this gives an alternative approximation to previous combinatorial
  algorithms \cite{Queyranne98,ZhaoNI05}. We raise the question as to whether
  the integrality gap is at most $1.5$.
\item There is an $O(\log n)$ for \SubML when the cut function is
  symmetric. We derive results for other special cases of \SubML.
\end{itemize}

\noindent
{\bf Rounding the convex relaxation:} Recall that the objective
function in \CR is $\sum_{i=1}^k \hat{f}_i(\bx_i)$, where
$\hat{f}_i(\bx_i) = \Ex_{\theta \in [0,1]}[f(\bx^\theta_i)]$. How do
we round while preserving the objective function?  If we focus on a
specific $i$, the objective function suggests that we pick $\theta$
randomly from $[0,1]$ and assign the elements in $\bx^\theta_i$ to
$i$; we call this $\theta$-rounding. However, there are two issues to
contend with. First, if we independently round for each $i$ then the
same element may be assigned multiple times. Second, we need to ensure
that all elements are assigned, which is not guaranteed by the
$\theta$-rounding. We remark that there is an integrality gap example
for hypergraph metric labeling that shows that there is no effective
rounding strategy that works in general.

Our approach is to understand the rounding process by considering
various special cases of interest. In particular, we consider monotone
functions, symmetric functions, the hypergraph separation cost
function (which is asymmetric), and combinations of such functions.
Monotonocity helps in that if elements are assigned to a label $i$,
they can be removed without increasing the fractional cost.  Although
one can use different strategies to obtain an $O(\log
n)$-approximation and integrality gap, a useful strategy here is the
rounding of Kleinberg and Tardos \cite{KleinbergT02} that they
introduced for metric labeling. This has the additional property of
ensuring that an element $u$ is assigned to $i$ with probability
exactly $x(u,i)$. We then consider the rounding process for \SubMP, in
particular the symmetric case \SymSubMP. Here, we crucially take
advantage of the fact that there is a single underlying function $f$,
and moreover the fact that it is symmetric. We consider the \CKR
strategy from \cite{CalinescuKR98} and show its effectiveness for
hypergraphs by abstracting away some of the properties specific to
graphs that were previously exploited in the analysis. In the process,
we also observe that a variant is equally effective for graphs but is
more insightful for \SymSubMP.

Finally, \SubML combines a monotone function and a non-monotone
function. Here, we resort to \KT since it is a reasonable strategy to
approximately preserve the cost of the monotone component. For the
uniform metric labeling problem, \cite{KleinbergT02} showed that \KT
approximately (to within a factor of $2$) preserves the fractional
connection cost in the case of graphs. We show bounds for hypergraph
cut functions in an analogous fashion.  Our insights enable us to
develop a variant of the rounding that gives an $O(\log
n)$-approximation for \SubML when the cut function is an arbitrary
symmetric submodular function.

\noindent
\textbf{Other Related Work:}
There has been much recent interest in optimizing with submodular set
functions. In particular, maximization problems have been examined via
combinatorial techniques as well as the multilinear relaxation
\cite{CalinescuCPV07}. The submodular welfare problem \cite{Vondrak08}
is similar in spirit to \MCSA except that one is interested in
maximizing the value of an allocation rather than minimizing the cost.
Minimization problems with submodular costs have also received
substantial attention
\cite{SvitkinaF08,GoemansHIM09,IwataN09,GoelKTW09} with several
negative results for basic problems as well as positive approximation
results for problems such as the submodular cost vertex cover problem
\cite{IwataN09,GoelKTW09}. Lov\'asz-extension based convex programs
have been effectively used for these problems.  Various submodular cut
and partition problems and their special cases such as the hypergraph
cut and partition have been studied recently
\cite{ZhaoNI05,Xiao10,OkumotoFN10,Fukunaga10}; however, these papers
have typically focussed on greedy and divide-and-conquer based
approaches while we use \CR.

\noindent
\textbf{Recent Results for \SymSubMP and \SubMP:}
Very recently,
building on the work in this paper and a non-trivial new technical
theorem, we showed \cite{ChekuriE11b} that the integrality gap of
\SubMPRel is at most $1.5 - 1/k$ for \SymSubMP and at most $2$ for
\SubMP.

\section{Monotone MSCA}
\label{sec:monotone}

In this section we consider \monMCSA where $f_1,\ldots,f_k$ are
monotone submodular functions. We will assume for simplicity that
$f_i(\emptyset) = 0$ for all $i$. Svitkina and Tardos
\cite{SvitkinaT06} considered this problem in the context of facility
location and gave a $(1+\ln n)$-approximation and matching hardness
via an approximation preserving reduction from set cover. Let $\alpha
= \min_{S \subseteq V, 1\le i \le k} f_i(S)/|S|$. The main observation
in \cite{SvitkinaT06} is that $\alpha \le \opt/n$, and moreover a pair
$(S,i)$ such that $f_i(S)/|S| = \alpha$ can be computed in
polynomial time via submodular function minimization. One can then
iterate using a greedy scheme, by using the monotonicity of the
functions, to obtain a $(1+\ln n)$-approximation.  Using a similar
argument, we can prove the following theorem.

\begin{theorem}
The integrality gap of \CR for \monMCSA is at most $(1+\ln n)$.
In particular, $\alpha \le \optcr/n$.
\end{theorem}

\begin{algo}
\underline{\textbf{\monGreedy:}}
\\\> let $x$ be a solution to \CR
\\\> $A_i \leftarrow \emptyset$ for all $i$ $(1 \leq i \leq k)$
\Comment{the set of vertices that will be assigned to $i$}
\\\> $U \leftarrow V$ \Comment{the set of unassigned vertices}
\\\> while $U \neq \emptyset$
\\\>\> let $\tilde{x}$ be the restriction of $x$ to $U$
\\\>\> for each $\theta$, let $A(i, \theta) = \{v \;|\; v \in
U, \tilde{x}(v, i) \geq \theta\}$
\\\>\> let $0 = \theta_{i, 1} < \theta_{i, 2} < \cdots <
\theta_{i, \ell_i} = 1$ be the distinct entries of $\tilde{\bx}_i$
\\\>\> let $(i', j')$ be the pair of indices in the set $\{(i, j)
\sep 1 \leq i \leq k, 1 \leq j < \ell_i\}$
\\\>\>\> that minimizes the ratio $f_i(A(i, \theta_{i, j})) / |A(i,
\theta_{i, j})|$
\\\>\> $A_{i'} \leftarrow A_{i'} \cup A(i', \theta_{i', j'})$
\\\>\> $U \leftarrow U - A(i', \theta_{i', j'})$
\end{algo}

\begin{theorem} \label{thm:monotone-greedy}
	\monGreedy achieves an $H_n$-approximation for \monMCSA.
\end{theorem}

\noindent
Before we prove Theorem~\ref{thm:monotone-greedy}, we introduce
some notation. Consider iteration of \monGreedy. Consider an iteration
of the while loop of \monGreedy. Let $U$ be the set of elements that
are unassigned at the beginning of the iteration, and let $\tilde{x}$
denote the restriction of $x$ to $U$; more precisely, $\tilde{x}(v, i)
= x(v, i)$ for all terminals $i$ and all vertices $v \in U$. For any
$\theta$, Let $A(i, \theta) = \{v \;|\; v \in U, \tilde{x}(v, i) \geq
\theta\}$.  Let $0 = \theta_{i, 1} < \theta_{i, 2} < \cdots <
\theta_{i, \ell_i} = 1$ be the distinct entries of $\tilde{\bx}_i$.
Let $\optcr = \sum_{i = 1}^k f_i(\bx_i)$.
Theorem~\ref{thm:monotone-greedy} follows from the following
lemma.

\begin{lemma} \label{lem:monotone-greedy-main}
	$$\min_{1 \leq i \leq k} \min_{0 \leq j < \ell_i} {f_i(A(i,
	\theta_{i, j})) \over |A(i, \theta_{i, j})|} \leq {\optcr \over
	|U|} .$$
\end{lemma}

\noindent
In order to prove Lemma~\ref{lem:monotone-greedy-main}, we will show
that, if we choose a terminal $i \in \{1, 2, \cdots, k\}$ and a
threshold $\theta \in [0, 1]$ uniformly at random, the ratio
$\Ex[f_i(A(i, \theta))] / \Ex[|A(i, \theta)|]$ is at most $\optcr /
|U|$. The following propositions give a bound on the two
expectations $\Ex[f_i(A(i, \theta))]$ and $\Ex[|A(i, \theta)|]$;
their proofs are relatively straightforward and they have been moved
to Appendix~\ref{app:monotone}.

\begin{prop} \label{prop:monotone-greedy-cost}
	$$\Ex_{i, \theta}[f_i(A(i, \theta))] \leq {1 \over k} \optcr.$$
\end{prop}
\begin{prop} \label{prop:monotone-greedy-size}
	$$\Ex_{i, \theta}[|A(i, \theta)|] = {1 \over k} |U|.$$
\end{prop}

\begin{proofof}{Lemma~\ref{lem:monotone-greedy-main}}
	Let $i$ be a terminal selected uniformly at random. Let $\theta$
	be a threshold selected uniformly at random from the interval $[0,
	1]$. It follows from Proposition~\ref{prop:monotone-greedy-cost}
	and Proposition~\ref{prop:monotone-greedy-size} that the ratio
	$\Ex[f_i(A(i, \theta))] / \Ex[|A(i, \theta)]$ is at most $\optcr /
	|U|$. By linearity of expectation,
		$$\Ex_{i, \theta}\Bigg[f_i(A(i, \theta)) - {\optcr \over |U|} \cdotp |A(i,
		\theta)|\Bigg] \leq 0$$
	and therefore there exists a terminal $i'$ and a threshold
	$\theta'$ for which the ratio $f_{i'}(A(i', \theta')) /
	|A(i', \theta')|$ is at most $\optcr / |U|$. Let $j'$ be the
	smallest index $j$ that satisfies $0 \leq j < \ell_{i'}$ and
	$\theta_{i', j} \geq \theta'$. Since $A(i', \theta_{i',
	j'}) = A(i', \theta')$, $(i', j')$ is the desired pair.
\end{proofof}

\noindent
In the remainder of this section, we consider a different algorithm
that achieves an $O(\log n)$-approximation for \monMCSA. We will use
this algorithm as a building block for submodular cost labeling
algorithms (see Section~\ref{sec:subml}).  The algorithm \KT is
derived from the work of Kleinberg and Tardos on metric labeling
\cite{KleinbergT02}.

\begin{algo}
\underline{\textbf{\KT}}
\\\> let $x$ be a solution to \CR
\\\> $S \leftarrow \emptyset$ \qquad \Comment{set of all assigned
vertices}
\\\> \Comment{set of vertices that are eventually assigned to $i$}
\\\> $A_i \leftarrow \emptyset$ for all $i$ $(1 \leq i \leq k)$
\\\> while $S \neq V$
\\\>\> pick $i \in \{1, 2, \cdots, k\}$ uniformly at random
\\\>\> pick $\theta \in [0, 1]$ uniformly at random
\\\>\> $A_i \leftarrow A_i \cup \left( \{v \;|\; x(v, i) \geq \theta\}
- S\right)$
\\\>\> $S \leftarrow S \cup A_i$
\\\> return $(A_1, \cdots, A_k)$
\end{algo}

\noindent
We prove the following theorem by building on some
useful properties that are shown in
\cite{KleinbergT02}; one of these is that the probability that
$v$ gets assigned to $i$ in the rounding is precisely $x(v,i)$. In
particular, this yields an optimal algorithm for modular functions.

\begin{theorem} \label{thm:monotone-kt-main}
  \KT achieves a randomized $O(\ln n)$-approximation for \MMCSA.
\end{theorem}
\begin{proofsketch}
	It is shown in \cite{KleinbergT02}, and not difficult to see, that
	the rounding terminates in $O(k \log n)$ iterations of the while
	loop with high probability. In each iteration the algorithm does a
	$\theta$-rounding on an index chosen uniformly at random. Let $i$
	be the random index and $A(i, \theta) = \{ v \mid x(v,i) \ge
	\theta \}$. Then it is easy to see that $\Ex[f_i(A(i, \theta))] =
	\sum_{i=1}^k \frac{1}{k} \hat{f}_i(\bx_i) = \frac{1}{k} \optcr$.
	Since the functions are monotone, we have that $\Ex[f(A(i, \theta)
	- S)] \le \frac{1}{k} \optcr$.  Since the algorithm terminates in
	$O(k \log n)$ iterations, by linearity of expectation and the
	sub-additivity of the functions (since the functions are
	submodular and $f(\emptyset) = 0$), the total expected cost is
	$O(\log n) \optcr$.
\end{proofsketch}

\section{Submodular Multiway Partition}
\label{sec:smp}

We consider \MCSA when the $f_i$ can be non-monotone.  We can show
that the integrality gap of \CR even for a special case of labeling on
hypergraphs can be $\Omega(n)$, and we suspect that the problem is
hard to approximate to a polynomial factor in $n$. We therefore focus
on \SubMPfull (\SubMP) and \SubMLfull (\SubML); these are broad
special cases which capture several problems that have been considered
previously.

The reduction of \SubMP to \MCSA requires one to work with the
non-terminals $V'$ as the ground set. It is however more convenient to
work with the terminals and non-terminals.  In particular, we work
with the relaxation below. Recall that $\bx_i = (x(v_1, i), \cdots,
x(v_n, i))$.

\begin{center}
\begin{boxedminipage}{0.55\linewidth}
\vspace{-0.2in}
\begin{align*}
& \textbf{\SubMPRel} &\\
\min \qquad  & \sum_{i = 1}^k \hat{f}(\bx_i) &\\
& \sum_{i = 1}^k x(v, i) &= 1 & \qquad \forall v\\
& x(s_i, i) &= 1 & \qquad \forall i\\
& x(v, i) &\geq 0 & \qquad \forall v, i
\end{align*}
\end{boxedminipage}
\end{center}

\noindent
As before, a starting point for rounding the relaxation is the basic
$\theta$-rounding that preserves the objective function. Suppose we do
$\theta$-rounding for each $i$ to obtain sets $A(1,
\theta),\ldots,A(k, \theta)$ where each $A(i, \theta) \subseteq V$.
Here we could use independent random $\theta$ values for each $i$ or
the same $\theta$. Note that the constraints ensure that $s_i \in A(j,
\theta)$ iff $i=j$. However, the sets $A(1, \theta), \ldots, A(k,
\theta)$ may intersect and also may not cover the entire set $V$, in
which case we have to allocate the remaining elements in some fashion.
First we show how to take advantage of the case when $f$ is symmetric
and then discuss how to obtain results for hypergraph problems that
are special cases of \SubMP.

\smallskip
\noindent
\textbf{A $2(1-1/k)$-approximation for \SymSubMP:} A
$2(1-1/k)$-approximation for \SymSubMP is known via greedy
combinatorial algorithms \cite{Queyranne98,ZhaoNI05}. However, no
mathematical programming formulation for the problem has been
previously considered. Here we show that, on instances of \SymSubMP,
the integrality gap of \CR is $2(1-1/k)$ by using an uncrossing
property of symmetric functions.

The following lemma is standard and it has been used in previous work
\cite{SvitkinaT04}.

\begin{lemma} \label{lem:uncross}
	Let $f$ be a symmetric submodular set function over $V$ and let
	$A_1,\ldots, A_k$ be subsets of $V$. Then there exist sets
	$A'_1,\ldots,A'_k$ such that (i) $A'_i \subseteq A_i$ for $1\le i\le
	k$, (ii) $A'_1,\ldots,A'_k$ are mutually disjoint (iii) $\cup_i A'_i
	= \cup_i A_i$ and (iv) $\sum_i f(A'_i) \le \sum_i f(A_i)$.
	Moreover, given the $A_i$'s a collection of sets $A'_i$ satisfying
	the above properties can be found in polynomial time via a value
	oracle for $f$.
\end{lemma}
\begin{proof}
	Since $f$ is symmetric, it satisfies posi-modularity; that is,
	$f(X) + f(Y) \ge f(X-Y) + f(Y-X)$. From this we see that either
	$f(X) + f(Y-X)$ or $f(Y) + f(X-Y)$ is no larger than $f(X) +
	f(Y)$. This allows us to uncross $A_1,\ldots,A_k$ as follows. If
	the $A_i$'s are mutually disjoint then we can set $A'_i = A_i$ for
	each $i$ and they satisfy the desired properties. Otherwise, there
	exist distinct $i$ and $j$ such that $A_i \cap A_j \neq
	\emptyset$. We can replace $A_i$ and $A_j$ with $A_i$ and $A_j -
	A_i$ if $f(A_i) + f(A_j-A_i) \le f(A_i)+f(A_j)$; otherwise, we
	replace them by $A_i-A_j$ and $A_j$. We repeat this process
	to get the desired sets.
\end{proof}

\begin{theorem}
	The integrality gap of \CR for \SymSubMP is $\le 2(1-1/k)$.
\end{theorem}
\begin{proof}
	Let $\bx$ be an optimal solution to \CR for a given instance of
	\SymSubMP. Let $A(1, \theta),\ldots, A(k, \theta)$ be sets
	obtained by applying $\theta$-rounding to each $i$. By the
	property of $\theta$-rounding, we observe that $\sum_i \Ex[f(A(i,
	\theta))] = \sum_i \hat{f}(\bx_i) = \optcr$.  Note that $s_i$
	belongs only to $A(i, \theta)$.  We now apply
	Lemma~\ref{lem:uncross} to $A(1, \theta),\ldots,A(k, \theta)$ to
	obtain $A'_1,\ldots,A'_k$. We have $\sum_i f(A'_i) \le \sum_i
	f(A(i, \theta))$. Let $V' = V - \cup_i A'_i$.  By symmetry of $f$,
	$f(V') = f(\cup_i A'_i)$ and, since $f$ is sub-additive, $f(V') =
	f(\cup_i A'_i) \le \sum_i f(A'_i) \le \sum_i f(A(i, \theta))$. We
	can allocate $V'$ to any index $i$ and the total cost of the
	allocation is upper bounded by $f(V') + \sum_i f(A'_i) \le 2\sum_i
	f(A(i, \theta))$. Thus the expected cost of the allocation is at
	most $2\optcr$. The allocation is feasible since $s_i$ belongs
	only to $A(i, \theta)$ and hence to $A'_i$.  One can refine this
	argument slightly to obtain a $2(1-1/k)$ bound; we give the
	details in Appendix~\ref{app:smp}.
\end{proof}

\smallskip
\noindent
In a previous version of this paper, we raised the following question.

\smallskip
\noindent
\textbf{Question.} Is the integrality gap of \CR for \SymSubMP at most $1.5$?

\noindent
As we already noted, we have shown in subsequent work
\cite{ChekuriE11b} that the integrality gap is at most $1.5 - 1/k$.

\medskip
\noindent
\textbf{Rounding for \AHMC and \HMP:} Calinescu \etal
\cite{CalinescuKR98} gave a new geometric relaxation for \MC, and a
rounding procedure that gave a $(1.5-1/k)$-approximation; the
integrality gap was subsequently improved to a bound of $1.3438
-\eps_k$ in \cite{KargerKSTY99}, while the best known lower bound is
$8/(7+1/k - 1)$ \cite{FreundK00}. Calinescu \etal \cite{CalinescuKR98}
derived their relaxation as a way to improve the integrality gap of
$2(1-1/k)$ for a natural distance based linear programming relaxation;
in fact, it often goes unnoticed that \cite{CalinescuKR98} shows the
equivalence of their geometric relaxation to that of another
relaxation obtained by adding valid strengthening constraints to the
distance based relaxation. Interestingly, when we specialize \MCSA to
\MC, \CR becomes the geometric relaxation of \cite{CalinescuKR98}! The
rounding procedure in \cite{CalinescuKR98} can be naturally extended
to rounding \CR for \SubMP and we describe it below.

\begin{algo}
\underline{\textbf{\CKR}}
\\\> let $x$ be a solution to \SubMPRel
\\\> pick a random permutation $\pi$ of $\{1, 2,
\cdots, k\}$
\\\> pick $\theta \in [0, 1)$ uniformly at random
\\\> $S \leftarrow \emptyset$ \qquad \Comment{set of all assigned
vertices}
\\\> for $i = 1$ to $k - 1$
\\\>\> $A_{\pi(i)} \leftarrow \left(\{v \;|\; x(v, \pi(i)) \geq
\theta\} - S\right)$
\\\>\> $S \leftarrow S \cup A_{\pi(i)}$
\\\> $A_{\pi(k)} \leftarrow V - S$
\\\> return $(A_1, \cdots, A_k)$
\end{algo}

\noindent
\CKR uses the same $\theta$ for all $i$ and a random permutation, both of
which are crucially used in the $1.5$-approximation analysis for \MC.
In this paper we investigate \CKR and other roundings for \AHMC and \HMP.

\noindent
Although \AHMC and \HMP appear similar, their objective functions are
different. The objective of \AHMC is to remove a minimum weight subset
of hyperedges such that the terminals are separated, whereas the
objective of \HMP is to minimize $\sum_e w(e) p(e)$, where $p(e)$ is
the number of non-trivial parts that $e$ is partitioned into (a part
is non-trivial if some vertex of $e$ is in that part but not all of
$e$). For graphs we have that either $p(e) = 0$ or $p(e) = 2$, and
therefore the two problems \AHMC and \HMP are equivalent; this is the
reason why one can view \MC as a partition problem as well. However,
when the hyperedges can have size larger than $2$, the objective
function values are not related to each other (it is easy to see that
the \HMP objective is always larger).

\HMP and \AHMC have been studied for their theoretical interest and
their applications. It is easy to see from its definition that \HMP is
a special case of \SymSubMP. It has been observed by a simple yet nice
reduction \cite{OkumotoFN10} that \AHMC is a special case of \SubMP.
In addition, it has been observed that \AHMC is
approximation-equivalent to the {\em node}-weighted multiway cut
problem in graphs (\nodeMC) \cite{GargVY04}.

We show that \CKR gives a $(1.5-1/k)$-approximation to \HMP and
a tight $H_\Delta$-approximation for \AHMC with maximum hyperedge
size $\Delta$. Note that when $\Delta=2$, $H_\Delta = 1.5$ and when
$\Delta=3$, $H_\Delta \simeq 1.833$. For $\Delta > 3$,
\CKR gives a worse than $2$ bound while we give an alternate rounding which
gives a $2(1-1/k)$-approximation. Our analysis of \CKR differs from
that in \cite{CalinescuKR98} since we cannot use the ``edge
alignment'' properties of the fractional solution that hold for graphs
and were exploited in \cite{CalinescuKR98}; our analysis for \AHMC is
inspired by the proof given by Williamson and Shmoys \cite{ShmoysW10}.

It is natural to wonder whether \CKR is crucial to obtaining a bound
that is better than $2$ for these problems, and in particular whether
it gives a $1.5$-approximation for \SymSubMP.  We show that a
$1.5-1/k$-approximation for \HMP (and hence \MC also) can be obtained
via a different algorithm as well; in particular, the crucial
ingredient in \CKR for \MC when viewed as a special case of \HMP is
the correlation provided by the use of the same $\theta$ for all $i$;
one can replace the random permutation by the uncrossing scheme in
Lemma~\ref{lem:uncross}. We describe this algorithm in the next
section. However, for \AHMC, the random permutation is important in
proving the $H_\Delta$-bound.

\subsection{A $1.5$-approximation for \HMPfull} \label{subsec:hmp}
We start by understanding the objective function of \SubMPRel in the
context of \HMP. Let $\bx$ be a feasible fractional solution, and let
$\bx_i = (x(v_1, i), \cdots, x(v_n, i))$ be the allocation to $i$.
Recall that $f$ here is the hypergraph cut function.  What is
$\sum_{i=1}^n \hat{f}(\bx_i)$?  For each terminal $i$ and each
hyperedge $e$, let $I(e, i) = [\min_{v \in e} x(v, i), \max_{v \in e}
x(v, i)]$. Let $d(e, i)$ denote the length of $I(e, i)$, and let $d(e)
= \sum_{i = 1}^k d(e, i)$. Note that $d(e) \in [0, |e|]$. 

\begin{lemma} \label{lem:sm-dist}
	$\sum_{i = 1}^k \hat{f}(\bx_i) = \sum_e w(e)d(e)$.
\end{lemma}
\begin{proof}
	Consider a hyperedge $e$.  Let $A(i, \theta)$ be the set whose
	characteristic vector is $\bx_i^{\theta}$. For each $\theta \in
	[0, \min_{v \in e} x(v, i)]$, the set $A(i, \theta)$ contains all
	the vertices of $e$, and thus $e \notin \delta(A(i, \theta))$.
	For each $\theta \in (\min_{v \in e} x(v, i), \max_{v \in e} x(v,
	i)]$, the set $A(i, \theta)$ contains at least one vertex of $e$
	but not all of the vertices of $e$, and thus $e \in \delta(A(i,
	\theta))$.  Finally, for each $\theta \in (\max_{v \in e} x(v, i),
	1]$, the set $A(i, \theta)$ does not contain any vertex of $e$,
	and thus $e \notin \delta(A(i, \theta))$. Therefore the
	contribution of $e$ to $\hat{f}(\bx_i)$ is equal to $(\max_{v
	\in e} x(v, i) - \min_{v \in e} x(v, i)) w(e) = d(e, i) w(e)$.
\end{proof}

\noindent
A crucial technical lemma that we need is the following which states
that the contribution of any $i$ to $d(e)$ is at most $d(e)/2$.

\begin{lemma} \label{lem:interval-distance}
	For any $i$, $d(e, i) \leq d(e) / 2$.
\end{lemma}
\begin{proof}
	Let $u = \mathrm{argmax}_{v \in e} x(v, i)$, and $w =
	\mathrm{argmin}_{v \in e} x(v, i)$. We have
	\begin{eqnarray*}
		d(e, i) &=& x(u, i) - x(w, i)\\
		&=& \left(1 - \sum_{j \neq i} x(u, j)\right) - \left(1 -
		\sum_{j \neq i} x(w, j)\right)\\
		&=& \sum_{j \neq i} (x(w, j) - x(u, j))\\
		&\leq& \sum_{j \neq i} \left(\max_{v \in e} x(v, j) - \min_{v
		\in e} x(v, j)\right)\\
		&=& \sum_{j \neq i} d(e, j)\\
		&=& d(e) - d(e, i)
	\end{eqnarray*}
	Therefore $d(e, i) \leq d(e) / 2$.
\end{proof}

\medskip
\noindent
The algorithm \SymMPR that we analyze is described below. We can prove
that \CKR gives the same bound; however, \SymMPR and its analysis are
perhaps more intuitive in the context of symmetric functions.  The
algorithm does $\theta$-rounding to obtain sets $A(1,
\theta),\ldots,A(k, \theta)$ and then uncrosses these sets to make
them disjoint without increasing the expected cost (see
Lemma~\ref{lem:uncross}).

\begin{algo}
\underline{\textbf{\SymMPR}}
\\\> let $x$ be a feasible solution to \SubMPRel
\\\> pick $\theta \in [0, 1]$ uniformly at random
\\\> $A(i, \theta) \leftarrow \{v \;|\; x(v, i) \geq \theta\}$ for
each $i$ $(1 \leq i \leq k)$
\\\> \Comment{uncross $A(1, \theta), \cdots, A(k, \theta)$}
\\\> $A'_i \leftarrow A(i, \theta)$ for each $i$ $(1 \leq i \leq k)$
\\\> while there exist $i \neq j$ such that $A'_i \cap A'_j \neq
\emptyset$
\\\>\> if $(f(A'_i) + f(A'_j - A'_i) \leq f(A'_i) + f(A'_j))$
\\\>\>\> $A'_j \leftarrow A'_j - A'_i$
\\\>\> else
\\\>\>\> $A'_i \leftarrow A'_i - A'_j$
\\\> return $(A'_1, \cdots, A'_{k - 1}, V - (A'_1 \cup \cdots
A'_{k - 1}))$
\end{algo}

\begin{theorem} \label{thm:hmp-main-thm}
	\SymMPR achieves an $1.5$-approximation for \HMP.
\end{theorem}

\noindent
We remark that we can change the algorithm and the analysis slightly
to achieve a $(1.5 - 1/k)$-approximation; we give the details in
Appendix~\ref{app:hmp}.

\begin{lemma} \label{lem:leftover-set}
	Let $i^*$ be the index such that the interval $I(e, i^*)$ has the
	rightmost ending point among the intervals $I(e, i)$. More
	precisely, $I(e, i^*)$ is an interval such that $\max_{v \in e}
	x(v, i^*) = \max_i \max_{v \in e} x(v, i)$; if there are several
	such intervals, we choose one arbitrarily. Let $Z_e$ be an
	indicator random variable equal to $1$ iff $e \in \delta(V - (A(1,
	\theta) \cup \cdots \cup A(k, \theta)))$. Then $\Ex[Z_e] \leq d(e,
	i^*)$.
\end{lemma}
\begin{proof}
	Note that $Z_e$ is equal to $1$ only if $(1)$ for any terminal
	$i$, $\theta$ is at least $\min_{v \in e} x(v, i)$ and $(2)$ there
	exists a terminal $\ell$ such that $\theta \in I(e, \ell)$. If
	there exists a terminal $i$ such that $\theta$ is smaller than
	$\min_{v \in e} x(v, i)$, all of the vertices of $e$ are in $A(i,
	\theta)$. If there does not exist a terminal $\ell$ such that
	$\theta \in I(e, \ell)$, either all of the vertices of $e$ are in
	$A(i, \theta)$ for some $i$ or all of the vertices of $e$ are in
	$V - (A(1, \theta) \cup \cdots \cup A(k, \theta))$.  Finally, we
	note that $(1)$ and $(2)$ imply that $\theta$ is in $I(e, i^*)$.
	By $(2)$, $\theta$ is at most $\max_{v \in e} x(v, i^*)$ and, by
	$(1)$, $\theta$ is at least $\min_{v \in e} x(v, i^*)$.
\end{proof}

\begin{proofof}{Theorem~\ref{thm:hmp-main-thm}}
	It follows from the sub-additivity of $f$ and
	Lemma~\ref{lem:uncross} that the cost of the partition returned by
	\SymMPR is at most
		$$\sum_{i = 1}^k f(A'_i) + f(V - (A'_1 \cup \cdots \cup A'_k))
		\leq \sum_{i = 1}^k f(A(i, \theta)) + f(V - (A(1, \theta) \cup
		\cdots \cup A(k, \theta)))$$
	Let $\optcr = \sum_{i = 1}^k \hat{f}_s(\bx_i)$.  By
	Lemma~\ref{lem:leftover-set} and
	Lemma~\ref{lem:interval-distance},
		$$\Ex[f(V - (A(1, \theta) \cup \cdots \cup A(k, \theta)))]
		\leq \sum_e {w(e)d(e) \over 2} = {\optcr \over 2}$$
	Finally, $\Ex[\sum_{i = 1}^k f(A(i, \theta))] = \optcr$, and
	therefore the expected cost of the allocation is at most $1.5
	\optcr$.
\end{proofof}

\subsection{Algorithms for \AHMCfull} \label{subsec:hmc}

Now we consider \AHMC. For each hyperedge $e$, pick an arbitrary
representative node $r(e) \in e$. Define the function $f:2^V
\rightarrow \mathbb{R}_+$ as follows: for $A \subseteq V$, let $f(A) =
\sum_{e: r(e) \in A, e \not \subseteq A} w(e)$ be the weight of
hyperedges whose representatives are in $A$ and they cross $A$. It is
easy to verify that $f$ is asymmetric and submodular. \SubMP with this
function $f$ captures \AHMC \cite{OkumotoFN10}.

Let $\bx$ be a feasible fractional allocation and $\bx_i$ be the
allocation for $i$. For each hyperedge $e$ and each terminal $i$, let
$I(e, i) = [\min_{v \in e} x(v, i), \max_{v \in e} x(v, i)]$.  Let
$d(e, i) = x(r(e), i) - \min_{v \in e} x(v, i)$ and $d(e) = \sum_{i =
1}^k d(e, i)$.

\begin{lemma}\label{lem:amc-distances}
	$\sum_{i = 1}^k \hat{f}(\bx_i) = \sum_e w(e)d(e)$.
\end{lemma}
\begin{proof}
	Consider a hyperedge $e$. Let $A(i, \theta)$ be the set whose
	characteristic vector is $\bx_i^{\theta}$. For each $\theta \in
	[0, \min_{v \in e} x(v, i)]$, the set $A(i, \theta)$ contains all
	the vertices of $e$ and therefore $e \notin \delta(A(i, \theta))$.
	For each $\theta \in (\min_{v \in e} x(v, i), x(r(e), i)]$, the
	set $A(i, \theta)$ contains the representative $r(e)$ of $e$ and
	$e \in \delta(A(i, \theta))$. Finally, for each $\theta \in
	(x(r(e), i), 1]$, the set $A(i, \theta)$ does not contain the
	representative $r(e)$. Therefore the contribution of $e$ to
	$f(\bx_i)$ is equal to $(x(r(e), i) - \min_{v \in e} x(v, i))
	w(e) = d(e, i) w(e)$.
\end{proof}

\begin{algo}
\underline{\textbf{\HR}}
\\\> let $x$ be a solution to \SubMPRel
\\\> pick $\theta \in (1/2, 1]$ uniformly at random
\\\> for $i = 1$ to $k - 1$
\\\>\> $A(i, \theta) \leftarrow \{v \;|\; x(v, i) \geq \theta\}$
\\\> $U(\theta) \leftarrow V - (A(1, \theta) \cup \cdots \cup A(k - 1,
\theta))$
\\\> return $(A(1, \theta), \cdots, A(k - 1, \theta), U(\theta))$
\end{algo}

\noindent
The main goal of this section is to show that \HR is a
$2$-approximation, and \CKR is a $H_{\Delta}$-approximation for \AHMC,
where $\Delta$ is the maximum hyperedge size. We remark that we can
show that the integrality gap of \CR is at most $2(1 - 1/k)$ using a
connection between the distance LP for \nodeMC considered in
\cite{GargVY04} and \CR for \AHMC; we give the details in
Appendix~\ref{app:hmc-improved-gap}. 

\begin{lemma} \label{lem:interval-sizes}
	Let $e$ be any hyperedge, and let $z$ be any vertex in $e$. Let
	$R(z) = \{i \;|\; x(z, i) = \max_{v \in e} x(v, i)\}$. Then
	$\sum_{i \in R(z)} |I(e, i)| \leq d(e)$.
\end{lemma}
\begin{proof}
	Let $u$ be the representative of $e$. If $z = u$, the lemma is
	immediate. Therefore we may assume that $z \neq u$. We partition
	$\{1, 2, \cdots, k\}$ into two sets: the set $S(z)$ consisting of
	all coordinates $i$ such that $x(z, i)$ is smaller than $x(u, i)$,
	and the set $B(z)$ consisting of all coordinates $i$ such that
	$x(z, i)$ is at least $x(u, i)$. Since $\sum_{i = 1}^k x(z, i)$
	and $\sum_{i = 1}^k x(u, i)$ are both equal to $1$, the total
	difference between $x(u, i)$ and $x(z, i)$ over all coordinates $i
	\in S(z)$ is equal to the total difference between $x(z, i)$ and
	$x(u, i)$ over all coordinates $i \in B(z)$. Therefore
	\begin{align*}
		\sum_{i \in B(z)}(x(z, i) - x(u, i)) &= \sum_{i \in S(z)}
		(x(u, i) - x(z, i))\\
		&\leq \sum_{i \in S(z)} \left(x(u, i) - \min_{v \in e} x(v,
		i)\right)\\
		&\leq d(e) - \sum_{i \in B(z)} \left(x(u, i) - \min_{v \in e}
		x(v, i)\right)
	\end{align*}
	Since $R(z)$ is a subset of $B(z)$, the lemma follows.
\end{proof}

\begin{corollary} \label{cor:interval-size}
	For each $i$, the length of the interval $I(e, i)$ is at most
	$d(e)$.
\end{corollary}
\begin{proof}
	Let $\beta = \max_{i = 1}^k \max_{v \in e} x(v, i)$. Let
	$s_{\ell}$ be a terminal and let $b$ be a vertex in $e$ such that
	$x(b, \ell) = \beta$. Since $\ell$ is in $R(b)$, the corollary
	follows from Lemma~\ref{lem:interval-sizes}.
\end{proof}

\begin{theorem} \label{thm:hmc-hr}
	Let $F$ be the set of all hyperedges crossing the partition
	returned by \HR. For each hyperedge $e$, $\Pr[e \in F] \leq 2
	d(e)$.
\end{theorem}
\begin{proof}
	Let $I(e, i^*)$ be the interval with the rightmost right interval
	among the intervals $I(e, 1), \cdots, I(e, k)$; if there are
	several such intervals, we pick one arbitrarily. Note that $e$ is
	in $F$ only if $\theta$ is in the interval $I(e, i^*)$.
	Therefore the probability that $e$ is in $F$ is at most $2 |I(e,
	i^*)|$. By Corollary~\ref{cor:interval-size}, the length of $I(e,
	i^*)$ is at most $d(e)$.
\end{proof}

\begin{theorem} \label{thm:hmc-ckr-cut-prob}
	Let $F$ be the set of all hyperedges crossing the partition
	returned by \CKR. For each hyperedge $e$, $\Pr[e \in F] \leq
	H_{|e|} \cdotp d(e)$.
\end{theorem}
\begin{proofsketch}
	We say that $s_i$ \emph{splits} $e$ if $\theta \in I(e, i)$. Let
	$X_i$ be the event that $s_i$ splits $e$. We say that $s_i$
	\emph{touches} $e$ if $\theta \leq \max_{v \in e} x(v, i)$.
	(Note that $\max_{v \in e} x(v, i)$ is the right endpoint of
	the interval $I(e, i)$.) We say that $s_i$ \emph{settles} $e$ if
	$s_i$ is the first terminal in the permutation $\pi$ that touches
	$e$. Let $Y_i$ be the event that $s_i$ settles $e$.

	Note that the edge $e$ is in $F$ only if there is a terminal $s_i$
	that splits \emph{and} settles $e$. Therefore we can upper bound
	the probability that $e$ is in $F$ by $\sum_{i = 1}^k \Pr[X_i
	\wedge Y_i]$.

	We relabel the terminals so that the ordering of the intervals
	$\{I(e, i)\}_{1 \leq i \leq k}$ from right to left according to
	their ending point is $I(e, 1), I(e, 2), \cdots, I(e, k)$. (If
	there are several intervals with the same ending point, we break
	ties arbitrarily.) After relabeling the intervals, $s_i$ settles
	$e$ only if, for each $j < i$, $\pi(s_i) < \pi(s_j)$. This
	observation, together with the fact that $s_i$ splits $e$ with
	probability $|I(e, i)|$, implies that $\Pr[X_i \wedge Y_i] \leq
	|I(e, i)| / i$.

	Finally, let $L(z) = \{i \;|\; x(z, i) = \max_{v \in e} x(v,
	i)\}$. If an index $i$ belongs to more than one set $L(z)$, we
	only add $i$ to one of the sets (chosen arbitrarily). Note that,
	by Lemma~\ref{lem:interval-sizes}, the total length of the
	intervals $I(e, i)$ where $i \in L(z)$ is at most $d(e)$. This,
	together with the fact that the sets $L(z)$ are disjoint and their
	union is $\{1, 2, \cdots, k\}$, implies that $\sum_{i = 1}^k
	\Pr[X_i \wedge Y_i]$ is at most $H_{|e|} \cdotp d(e)$.
\end{proofsketch}

\begin{prop} \label{prop:ahmc-ckr-tight}
	The analysis in Theorem~\ref{thm:hmc-ckr-cut-prob} is tight.
\end{prop}
\begin{proof}
	Let $e$ be a hyperedge with representative $u$. Let $\epsilon \in
	(0, 1)$ be such that $\epsilon |e| \leq 1$. Consider a solution
	$x$ that assigns the following values to the vertices of $e$. For
	each terminal $i > |e|$ and each vertex $z \in e$, we have $x(z,
	i) = 0$.  For each terminal $i$ such that $1 < i \leq |e|$, we
	have $x(u, i) = (|e| - i) \epsilon$. Finally, $x(u, 1) = 1 -
	\sum_{i = 2}^k x(u, i)$.  Let $v_2, \cdots, v_{|e|}$ denote the
	remaining vertices of $e$ (other than $u$). Now consider an index
	$j$ such that $2 \leq j \leq |e|$. We have $x(v_j, 1) = \epsilon$,
	$x(v_j, j) = x(u, j) - \epsilon$, and $x(v_j, i) = x(u, i)$ for
	all $i \neq j$. Note that $\sum_{i = 1}^k x(v, i)$ is equal to $1$
	for all vertices $v \in e$ and the distance $d(e)$ is equal to
	$\epsilon$. It is straightforward to verify that $e$ is in $F$
	with probability at least $H_{|e|} \cdotp d(e)$.
\end{proof}

\section{Submodular Cost Labeling} \label{sec:subml}
In this section we consider \SubML, which generalizes \monMCSA,
uniform metric labeling, hub location, and other problems. A natural
algorithm here is \KT, which we have already introduced in
Section~\ref{sec:monotone}. We also describe a different algorithm,
\SymMLR, which is appropriate for \SubML when the cut function is an
arbitrary symmetric submodular function. We obtain several results
that we state below.

The next two results consider the \SubML problem on hypergraphs in
which $h$ is the following function. For each edge hyperedge $e$, pick
an arbitrary representative node $r(e) \in e$. We define the function
$h:2^V \rightarrow \mathbb{R}_+$ as follows: for $A \subseteq V$, let
$f(A) = \sum_{e: r(e) \in A, e \not \subseteq A} w(e)$ be the weight
of hyperedges whose representatives are in $A$ and they cross $A$. We
refer to this function as the hypergraph separation cost function.

\begin{theorem} \label{thm:uml-hypercut}
	If $h$ is the hypergraph separation cost function and each $g_i$
	is modular, \KT achieves a $\Delta$-approximation for \SubML.
\end{theorem}

\begin{theorem} \label{thm:ml-hypercut}
	If $h$ is the hypergraph separation cost function and each $g_i$
	is a monotone submodular function, \KT achieves an $O(\ln n +
	\Delta)$ approximation for \SubML.
\end{theorem}

\begin{algo}
\underline{\textbf{\SymMLR}}
\\\> let $x$ be a solution to \CR
\\\> $S \leftarrow \emptyset$ \qquad \Comment{set of all assigned
vertices}
\\\> $A_i \leftarrow \emptyset$ for all $i$ $(1 \leq i \leq k)$
\\\> while $S \neq V$
\\\>\> pick $i \in \{1, 2, \cdots, k\}$ uniformly at random
\\\>\> pick $\theta \in [0, 1]$ uniformly at random
\\\>\> $A_i \leftarrow A_i \cup \{v \;|\; x(v, i) \geq \theta\}$
\\\>\> $S \leftarrow S \cup \{v \;|\; x(v, i) \geq \theta\}$
\\\> \Comment{uncross $A_1, \cdots, A_k$}
\\\> $A'_i \leftarrow A_i$ for all $i$ $(1 \leq i \leq k)$
\\\> while there exist $i \neq j$ such that $A'_i \cap A'_j \neq
\emptyset$
\\\>\> if $(f(A'_i) + f(A'_j - A'_i) \leq f(A'_i) + f(A'_j))$
\\\>\>\> $A'_j \leftarrow A'_j - A'_i$
\\\>\> else
\\\>\>\> $A'_i \leftarrow A'_i - A'_j$
\\\> return $(A'_1, \cdots, A'_k)$
\end{algo}

\begin{theorem} \label{thm:ml-symmetric}
	If $h$ is a symmetric submodular function and each $g_i$ is a
	monotone submodular function, \SymMLR achieves an $O(\ln
	n)$ approximation for \SubML.
\end{theorem}

\noindent
Let $\partcr = \sum_{i = 1}^k \hat{h}(\bx_i)$ be the partition cost of
\CR, and let $\costcr = \sum_{i = 1}^k \hat{g}_i(\bx_i)$ be the
assignment cost of \CR.

Consider the \SubML problem in which $h$ is the hypergraph separation
cost function. For each hyperedge $e$, let $d(e) = \sum_{i = 1}^k
(x(r(e), i) - \min_{v \in e} x(v, i))$. By
Lemma~\ref{lem:amc-distances}, $\partcr = \sum_e w(e) d(e)$. In order
to bound the expected partition cost of the labeling constructed by
\KT, we consider each hyperedge separately, and we give an upper bound
on the probability that the hyperedge has at least two vertices with
different labels. We say that a hyperedge $e$ is split in some
iteration of \KT if there exists an iteration $\ell$ such that at
least one vertex of $e$ is assigned a label in iteration $\ell$ but
not all vertices of $e$ are assigned a label in iteration $\ell$. The
following lemma gives an upper bound on the probability that a
hyperedge $e$ is split.

\begin{lemma}\label{lem:kt-split-ahmc}
	For each hyperedge $e$, the probability that $e$ is split is at
	most $\Delta d(e)$.
\end{lemma}

\noindent
Using Lemma~\ref{lem:kt-split-ahmc}, we can complete the proofs of
Theorem~\ref{thm:uml-hypercut} and Theorem~\ref{thm:ml-hypercut} as
follows.

\begin{proofof}{Theorem~\ref{thm:uml-hypercut}}
	Let $(A_1, \cdots, A_k)$ be the partition returned by \KT, and let
	$\partint = \sum_{i = 1}^k h(A_i)$ and $\costint = \sum_{i = 1}^k
	g_i(A_i)$.  As shown in \cite{KleinbergT02}, \KT assigns label $i$
	to $v$ with probability $x(v, i)$. Thus $\Ex[\costint] = \costcr$.

	By Lemma~\ref{lem:amc-distances}, $\partcr = \sum_e w(e) d(e)$.
	Therefore, by Lemma~\ref{lem:kt-split-ahmc}, $\Ex[\partint] \leq
	\Delta \cdotp \partcr$.
\end{proofof}

\begin{proofof}{Theorem~\ref{thm:ml-hypercut}}
	Let $(A_1, \cdots, A_k)$ be the partition returned by \KT, and let
	$\partint = \sum_{i = 1}^k h(A_i)$ and $\costint = \sum_{i = 1}^k
	g_i(A_i)$. Using the argument in the proof of
	Theorem~\ref{thm:monotone-kt-main}, we can show that
	$\Ex[\costint] \leq O(\ln n) \costcr$. Additionally, by
	Lemma~\ref{lem:kt-split-ahmc}, $\Ex[\partint] \leq \Delta \partcr$.
\end{proofof}

\smallskip
\noindent
Now we turn our attention to the proof of
Lemma~\ref{lem:kt-split-ahmc}.

\begin{proofof}{Lemma~\ref{lem:kt-split-ahmc}}
	Consider iteration $\ell$ of \KT, and let $i_{\ell}$ and
	$\theta_{\ell}$ be the label and threshold in iteration $\ell$.
	We say that iteration $\ell$ \emph{cuts} $e$ if $\theta_{\ell} \in
	[\min_{v \in e} x(v, i_{\ell}), \max_{v \in e} x(v, i_{\ell})]$.
	We say that iteration $\ell$ \emph{touches} $e$ if $\theta_{\ell}$
	is in the interval $[0, \max_{v \in e} x(v, i_{\ell})]$. Let
	$X_{\ell}$ and $Z_{\ell}$ be the events that $\ell$ cuts and
	touches $e$ (respectively).  The probability that $e$ is split in
	iteration $\ell$ is at most $\Pr[X_{\ell}] / \Pr[Z_{\ell}]$. We
	have
		$$\Pr[X_{\ell}] \leq {1 \over k} \sum_{i = 1}^k \left(\max_{v
		\in e} x(v, i) - \min_{v \in e} x(v, i)\right) \leq {\Delta
		d(e) \over k}$$
	where the last inequality follows from
	Lemma~\ref{lem:interval-sizes}. Additionally,
		$$\Pr[Z_{\ell}] = {1 \over k} \sum_{i = 1}^k \max_{v \in e}
		x(v, i) \geq {1 \over k}$$
	The last inequality follows from the fact that, for any vertex $w
	\in e$, $\sum_{i = 1}^k x(w, i) = 1$. It follows that the
	probability that $e$ is split in iteration $j$ is at most $\Delta
	d(e)$.
\end{proofof}

\begin{proofof}{Theorem~\ref{thm:ml-symmetric}}
	Let $i_{\ell}$ and $\theta_{\ell}$ be the label and $\theta$ value
	chosen in the $\ell$-th iteration of the \emph{first} while loop
	\SymMLR. For each $i$, let $\script{A}_i = \cup_{\ell:\; i_{\ell}
	= i} \{v \sep x(v, i) \geq \theta_{\ell}\}$. Let $\partballs =
	\sum_{i = 1}^k f(\script{A}_i)$ and $\costballs = \sum_{i = 1}^k
	g_i(\script{A}_i)$. (Note that $\script{A}_i$ is the set $A_i$ at
	the end of the \emph{first} while loop of \SymMLR.)

	Using the argument in the proof of
	Theorem~\ref{thm:monotone-kt-main}, we can show that
		$$\Ex[\costballs] \leq O(\ln n) \costcr$$
	and
		$$\Ex[\partballs] \leq O(\ln n) \partcr$$
	Let $(A'_1, \cdots, A'_k)$ be the partition returned by \SymMLR.
	Let $\partint = \sum_{i = 1}^k f(A'_i)$ and $\costint = \sum_{i =
	1}^k g_i(A'_i)$. By Lemma~\ref{lem:uncross},
		$$\Ex[\partint] \leq \Ex[\partballs] \leq O(\ln n) \partcr$$
	Since each $g_i$ is
	monotone,
		$$\Ex[\costint] \leq \Ex[\costballs] \leq O(\ln n) \costcr.$$
\end{proofof}

\smallskip
\noindent
\textbf{Integrality Gap Example:}
We remark that we can generalize the integrality gap example of
\cite{KleinbergT02} in order to show that the integrality gap of \CR
is at least $\Delta(1 - 1/k)$ for \SubML when $h$ is the hypergraph
separation cost function, even if each $g_i$ is modular.

Consider a $\Delta$-uniform complete hypergraph on $k$ vertices; all
${k \choose \Delta}$ edges are present, and each edge has unit weight.
For each vertex $i$, the cost of assigning label $j$ to $i$ is zero if
$i \neq j$, and infinity otherwise.

It is easy to see that an optimal integral solution picks a label $i$
and assigns label $i$ to all vertices except $i$, and it assigns some
other label to $i$. Thus the integral optimum is ${k - 1 \choose
\Delta - 1}$.  Setting $x(i, j) = 1 / (k - 1)$ for all $i \neq j$
gives us a fractional solution of cost ${k \choose \Delta} / (k - 1)$.

\bigskip
\noindent
\textbf{Acknowledgments:} We thank Lisa Fleischer for suggesting that
we contact Zoya Svitkina about \MCSA and thank Zoya for pointing out
her work in \cite{SvitkinaT06} on monotone \MCSA. CC thanks Jan Vondrak
for pointing out the interpretation of the Lov\'asz extension from his
paper \cite{Vondrak09} which was very helpful in thinking about
rounding procedures. AE thanks Sungjin Im and Ben Moseley for 
discussions.
\newpage
\bibliographystyle{plain}

\newpage
\appendix

\section{Definition of the Lov\'asz extension}
\label{app:lovasz}
Let $f: \{0, 1\}^n \rightarrow \mathbb{R}$ be a function. The Lov\'asz
extension $\hat{f}$ of $f$ is the function $\hat{f}: [0, 1]^n
\rightarrow \mathbb{R}$ defined as follows. Let $\bx$ be a vector in
$[0, 1]^n$. We relabel the vertices as $1, 2, \cdots, n$ so that
$x_1 \geq x_2 \geq \cdots \geq x_n$; for ease of notation,
let $x_0 = 1$ and $x_{n + 1} = 0$. Let $S_i = \{1, 2, \cdots, i\}$.
The value of $\hat{f}$ at $\bx$ is equal to
	$$\hat{f}(\bx) = \sum_{i = 0}^{n} (x_{i + 1} - x_i)
	f(S_i)$$
It is straightforward to verify that $\sum_{i = 0}^{n} (x_{i + 1} -
x_i) f(S_i) = \Ex_{\theta \in [0, 1]}[f(\bx^{\theta})]$.

Another useful extension for a function $f$ is its convex closure,
which is defined as follows. For each set $S \subseteq V$, we let
$\b1_S$ denote the characteristic vector of $S$; that is, the $i$-th
coordinate of $\b1_S$ is equal to $1$ if $i$ is in $S$ and $0$
otherwise. The convex closure $f$ is the function $f^-: [0,
1]^n \rightarrow \mathbb{R}$ where $f^-(\bx) = \min\{\sum_{S
  \subseteq V} \lambda_S f(S) \;:\; \sum_{S \subseteq V} \lambda_S
\b1_S = \bx, \sum_{S \subseteq V} \lambda_S = 1, \lambda_S \geq
0\}$. The Lov\'asz extension $\hat{f}$ of $f$ is equal to the convex
closure $f^-$ of $f$ iff $f$ is submodular; see for instance
\cite{Dughmi09}.  Using this result, we can show that \CR can be
solved in time that is polynomial in $n$ and $\log\left(\max_{i, S
    \subseteq V} f_i(S)\right)$ via the ellipsoid method.

Since $\hat{f}$ is equal to $\hat{f}^-$, we can write \CR as follows.

\begin{center}
\begin{boxedminipage}{0.5\textwidth}
\vspace{-0.2in}
\begin{align*}
& \textbf{\CR-Primal} &\\
\min \qquad & \sum_{i = 1}^k \sum_{S \subseteq V} \lambda(S, i)
f_i(S) &\\
& \sum_{S: v \in S} \lambda(S, i) = x(v, i) & \forall v, i\\
& \sum_{S \subseteq V} \lambda(S, i) = 1 & \forall i\\
& \sum_{i = 1}^k x(v, i) = 1 & \forall v\\
& \lambda(S, i) \geq 0 & \forall S, i\\
& x(v,i) \geq 0  & \forall v, i
\end{align*}
\end{boxedminipage}
\end{center}

\noindent
\textbf{\CR-Primal} is an LP with exponentially many variables and
polynomially many constraints. Its dual \textbf{\CR-Dual} has
polynomially many variables and exponentially many constraints.

\begin{center}
\begin{boxedminipage}{0.5\textwidth}
\vspace{-0.2in}
\begin{align*}
& \textbf{\CR-Dual} &\\
\max \qquad & \sum_{i = 1}^k \beta_i + \sum_{v \in V} \gamma_v &\\
& \sum_{v \in S} \alpha(v, i) + \beta_i \leq f_i(S) & \forall S, i\\
& \gamma_v \leq \alpha(v, i) & \forall v, i
\end{align*}
\end{boxedminipage}
\end{center}

\noindent
\textbf{Separation oracle for \CR-Dual.} Fix an assignment of values
to the variables $\alpha, \beta, \gamma$ in \textbf{\CR-Dual}.  It is
easy to check in polynomial time whether $\gamma_v \le \alpha(v,i)$
for all $v,i$ since there are only $nk$ such constraints.  Let $g_i(S)
= \sum_{v \in S} \alpha(v, i) + \beta_i$. Note that $g_i$ is a
\emph{modular} function and therefore $f_i - g_i$ is a submodular
function. Using a polynomial time algorithm for submodular function
minimization, for a given $i$, we can check whether $f_i(S) - g_i(S)
\geq 0$ for all sets $S \subseteq V$.

Therefore we can solve \textbf{\CR-Dual} in time that is polynomial in
$n$ and $\log\left(\max_{i, S\subseteq V} f_i(S)\right)$ using the
ellipsoid method. Using standard techniques, we can also construct an
optimal solution for the primal; we omit the details here.

\section{Omitted proofs from Section~\ref{sec:monotone}}
\label{app:monotone}

\begin{proofof}{Proposition~\ref{prop:monotone-greedy-cost}}
	Consider a terminal $i$. For any $\theta$, $A(i, \theta)$ is a
	subset of $\{v \sep v \in V, x(v, i) \geq \theta\}$. Since $f_i$
	is monotone, $f_i(A(i, \theta)) \leq f_i(\{v \sep v \in V, x(v, i)
	\geq \theta\})$. Therefore
	\begin{eqnarray*}
		\Ex_{\theta \in [0, 1]}[A(i, \theta)] &=& \int_0^1 f_i(A(i,
		\theta)) d\theta\\
		&\leq& \int_0^1 f_i(\{v \sep v \in V, x(v, i) \geq \theta\})
		d\theta\\
		&=& \hat{f}_i(\bx_i)
	\end{eqnarray*}
	Finally,
		$$\Ex_{i, \theta}[f_i(A(i, \theta))] = {1 \over k} \sum_{i =
		1}^k \Ex_{\theta \in [0, 1]}[f_i(A(i, \theta)] \leq {1 \over
		k} \optcr$$
\end{proofof}

\begin{proofof}{Proposition~\ref{prop:monotone-greedy-size}}
	Note that
		$$\Ex_{i, \theta}[|A(i, \theta)|] = {1 \over k} \sum_{i = 1}^k
		\Ex_{\theta \in [0, 1]}[|A(i, \theta)|]$$
	We can prove by induction on the size of $U$ that $\sum_{i = 1}^k
	\Ex_{\theta \in [0, 1]}[|A(i, \theta)|]$ is equal to $|U|$.

	If $U$ is empty, the claim trivially holds.  Therefore we
	may assume that $U$ contains at least one element $z$. Let $U' = U
	- \{z\}$, and let $\tilde{x}'$ be the restriction of $\tilde{x}$
	to $U'$; more precisely, $\tilde{x}'(v, i) = \tilde{x}(v, i)$ for
	all $v \in U'$ and all terminals $i$. Let $A'(i, \theta) = \{v
	\;|\; v \in U', \tilde{x}'(v, i) \geq \theta\}$. Note that
	$A'(i, \theta)$ is equal to $A(i, \theta) - \{z\}$ if $\theta$ is
	smaller than $x(z, i)$, and $A'(i, \theta)$ is equal to
	$A(i, \theta)$ otherwise. Therefore
	\begin{eqnarray*}
		\sum_{i = 1}^k \Ex_{\theta \in [0, 1]}[|A(i, \theta)|] &=&
		\sum_{i = 1}^k \int_0^1 |A(i, \theta)| d\theta\\
		&=& \sum_{i = 1}^k \int_0^{x(z, i)} |A'(i, \theta) \cup \{z\}|
		d\theta + \sum_{i = 1}^k \int_{x(z, i)}^1 |A'(i, \theta)|
		d\theta\\
		&=& \sum_{i = 1}^k \int_0^1 |A'(i, \theta)| d\theta + \sum_{i
		= 1}^k x(z, i)\\
		&=& \sum_{i = 1}^k \int_0^1 |A'(i, \theta)| d\theta + 1\\
		&=& |U'|+ 1 \qquad\qquad \mbox{(By induction)}\\
		&=& |U|
	\end{eqnarray*}
\end{proofof}

\section{Omitted proofs from Section~\ref{sec:smp}}
\label{app:smp}

\begin{theorem}
	The integrality gap of \CR for \SymSubMP is at most $2(1-1/k)$.
\end{theorem}
\begin{proof}
	Let $\bx$ be an optimal solution to \CR for a given instance of
	\SymSubMP. Without loss of generality, $\hat{f}(\bx_k) = \max_i
	\hat{f}(\bx_i)$.  Let $A(1, \theta),\ldots, A(k - 1, \theta)$ be
	sets obtained by applying $\theta$-rounding with respect to the
	first $k - 1$ terminals. By the property of $\theta$-rounding, we
	observe that $\sum_{i = 1}^{k - 1} \Ex[f(A(i, \theta))] = \sum_{i
	= 1}^{k - 1} \hat{f}(\bx_i) \leq (1 - 1/k) \optcr$.  The last
	inequality follows from the fact that $\hat{f}(\bx_k) \geq \optcr
	/ k$.  Note that $s_i$ belongs only to $A(i, \theta)$.  We now
	apply Lemma~\ref{lem:uncross} to $A(1, \theta),\ldots,A(k - 1,
	\theta)$ to obtain $A'_1,\ldots,A'_{k - 1}$. We have $\sum_{i =
	1}^{k - 1} f(A'_i) \le \sum_{i = 1}^{k - 1} f(A(i, \theta))$. Let
	$A'_k = V - \cup_{1 \leq i \leq k - 1} A'_i$.  By symmetry of $f$,
	$f(A'_k) = f(\cup_{1 \leq i \leq k - 1} A'_i)$ and, since $f$ is
	sub-additive, $f(A'_k) = f(\cup_{1 \leq i \leq k - 1} A'_i) \le
	\sum_{i = 1}^{k - 1} f(A'_i) \le \sum_{i = 1}^{k - 1} f(A(i,
	\theta))$. We allocate $A'_k$ to index $k$ and the total cost of
	the allocation is upper bounded by $\sum_i f(A'_i) \le 2\sum_{i =
	1}^{k - 1} f(A(i, \theta))$.  Thus the expected cost of the
	allocation is at most $2(1 - 1/k)\optcr$. The allocation is
	feasible since, for each $i \neq k$, $s_i$ belongs only to $A(i,
	\theta)$ and hence to $A'_i$, and $s_k$ belongs to $A'_k$.
\end{proof}

\section{Omitted proofs from Subsection~\ref{subsec:hmp}}
\label{app:hmp}

\begin{algo}
\underline{\textbf{\SymMPR}}
\\\> let $x$ be a feasible solution to \SubMPRel
\\\> relabel the terminals so that $\hat{f}(\bx_k) = \max_i
\hat{f}(\bx_i)$
\\\> pick $\theta \in [0, 1]$ uniformly at random
\\\> $A(i, \theta) \leftarrow \{v \;|\; x(v, i) \geq \theta\}$ for
each $i$ $(1 \leq i \leq k - 1)$
\\\> \Comment{uncross $A(1, \theta), \cdots, A(k - 1, \theta)$}
\\\> $A'_i \leftarrow A(i, \theta)$ for each $i$ $(1 \leq i \leq k - 1)$
\\\> while there exist $i \neq j$ such that $A'_i \cap A'_j \neq
\emptyset$
\\\>\> if $(f(A'_i) + f(A'_j - A'_i) \leq f(A'_i) + f(A'_j)$
\\\>\>\> $A'_j \leftarrow A'_j - A'_i$
\\\>\> else
\\\>\>\> $A'_i \leftarrow A'_i - A'_j$
\\\> return $(A'_1, \cdots, A'_{k - 1}, V - (A'_1 \cup \cdots
A'_{k - 1}))$
\end{algo}

\begin{theorem} \label{thm:hmp}
	\SymMPR achieves an $(1.5 - 1/k)$-approximation for \HMP.
\end{theorem}

\begin{lemma} \label{lem:leftover-set2}
	Let $i^*$ be the index such that the interval $I(e, i^*)$ has the
	rightmost ending point among the intervals $I(e, i)$, where $1
	\leq i \leq k - 1$. More precisely, $I(e, i^*)$ is an interval
	such that $\max_{v \in e} x(v, i^*) = \max_{1 \leq i \leq k - 1}
	\max_{v \in e} x(v, i)$; if there are several such intervals, we
	choose one arbitrarily. Let $Z_e$ be an indicator random variable
	equal to $1$ iff $e \in \delta(V - (A(1, \theta) \cup \cdots \cup
	A(k - 1, \theta)))$. Then $\Ex[Z_e] \leq d(e, i^*)$.
\end{lemma}
\begin{proof}
	Note that $Z_e$ is equal to $1$ only if $(1)$ for any terminal $i
	\neq k$, $\theta$ is at least $\min_{v \in e} x(v, i)$ and $(2)$
	there exists a terminal $\ell \neq k$ such that $\theta \in I(e,
	\ell)$. If there exists a terminal $i \neq k$ such that $\theta$
	is smaller than $\min_{v \in e} x(v, i)$, all of the vertices of
	$e$ are in $A(i, \theta)$. If there does not exist a terminal
	$\ell \neq k$ such that $\theta \in I(e, \ell)$, either all of the
	vertices of $e$ are in $A(i, \theta)$ for some $i \neq k$ or all
	of the vertices of $e$ are in $V - (A(1, \theta) \cup \cdots \cup
	A(k - 1, \theta))$.  Finally, we note that $(1)$ and $(2)$ imply
	that $\theta$ is in $I(e, i^*)$: by $(2)$, $\theta$ is at most
	$\max_{v \in e} x(v, i^*)$ and, by $(1)$, $\theta$ is at least
	$\min_{v \in e} x(v, i^*)$.
\end{proof}

\begin{proofof}{Theorem~\ref{thm:hmp}}
	It follows from Lemma~\ref{lem:uncross} that the cost of the
	partition returned by \SymMPR is at most
		$$\sum_{i = 1}^{k - 1} f(A'_i) + f(V - (A'_1 \cup \cdots \cup
		A'_{k - 1}))
		\leq \sum_{i = 1}^{k - 1} f(A(i, \theta)) + f(V - (A(1, \theta) \cup
		\cdots \cup A(k - 1, \theta)))$$
	By Lemma~\ref{lem:leftover-set2} and
	Lemma~\ref{lem:interval-distance},
		$$\Ex[f(V - (A(1, \theta) \cup \cdots \cup A(k - 1, \theta)))]
		\leq \sum_e {w(e)d(e) \over 2} = {\optcr \over 2}$$
	Since $\hat{f}(\bx_k) \geq \optcr / k$, we have
		$$\Ex\Bigg[\sum_{i = 1}^{k - 1} f(A(i, \theta))\Bigg] \leq
		\left(1 - {1 \over k}\right) \optcr$$
	Therefore the expected cost of the allocation is at most $(1.5 -
	1/k)\optcr$.
\end{proofof}

\section{Improved integrality gap bound for \AHMC}
\label{app:hmc-improved-gap}
The main goal of this section is to establish a connection between the
distance LP for \nodeMC considered in \cite{GargVY04} and \CR for
\AHMC. This connection gives us the following theorem.

\begin{theorem}
	The integrality gaph of \AHMC is at most $2(1 - 1/k)$ for instances
	of \AHMC.
\end{theorem}

\noindent
The reader may wonder whether we can prove the above theorem directly
without recourse to the result from \cite{GargVY04}; we believe that it
can be done but there are some technical hurdles that we plan to 
address in a future version of the paper.
As we already noted, \AHMC and \nodeMC are approximation equivalent.
Okumoto \etal \cite{OkumotoFN10} gave an approximation-preserving
reduction from \AHMC to \nodeMC.  The reduction maps an instance of
\AHMC to an instance of \nodeMC as follows. Let $G = (V, E)$ be an
instance of \AHMC, namely a hypergraph with weights on the edges and
$k$ distinguished vertices which we call terminals. We construct a
bipartite graph $H$ as follows.  We add all of the vertices of $G$ to
$H$, and we give them infinite weight. For each hyperedge $e$ of $G$,
we add a vertex to $H$ of weight $w_e$, and we connect it to all of
vertices of $V$ that are contained in the hyperedge. The terminals of
$H$ are the terminals of $G$, and it is straightforward to verify that
a multiway cut in $G$ corresponds to a node multiway cut in $H$ of the
same weight and vice-versa.

We remark that there is an approximation-preserving reduction from
\nodeMC to \AHMC as well. Let $G = (V, E)$ be an instance of \nodeMC,
namely a graph with weights on the vertices and $k$ distinguished
vertices called terminals. We may assume without loss of generality
that the terminals form an independent set of $G$. We subdivide each
edge of $G$, except the edges incident to the terminals. Now we can
view the resulting graph $G'$ as a bipartite graph with the terminals
and the subdividing vertices on the left, and all other vertices on
the right. We construct a hypergraph $H$ as follows: the vertices of
$H$ are the left vertices of $G'$ and, for each vertex $v$ on the
right, $H$ has a hyperedge consisting of all neighbors of $v$ (in the
subdivided graph $G'$). The weight of the hyperedge corresponding to
$v$ is $w(v)$.  The terminals of $H$ are the terminals of $G$, and it
is straightforward to verify that a node multiway cut in $G$
corresponds to a multiway cut in $H$ of the same weight and vice-versa.

Consider an instance $G$ of \AHMC, and let $\bx$ be a feasible solution
to \CR for this instance.  Using the first reduction, we map an
instance $G$ of \AHMC to the instance $H$ of \nodeMC described above.
In the following, we show that we can map the solution $\bx$ to a
solution $\bd$ to the distance LP relaxation for $H$. The distance LP
has a variable $d_v$ for each non-terminal $v$ with the interpretation
that $d_v$ is $1$ if $v$ is in the node multiway cut. Let
$\script{P}_{s_i, s_j}$ denote the set of all paths of $H$ from $s_i$
to $s_j$.

\begin{center}
\begin{boxedminipage}{0.5\textwidth}
\begin{align*}
\vspace{-0.2in}
(\textsc{Distance-LP})\qquad \min& \;\; \sum_v d_v w_v\\
& \sum_{v \in p} d_v &\geq 1 & \qquad \forall i \neq j, \forall p \in
\script{P}_{s_i, s_j}\\
& d_v &\geq 0 & \qquad \forall v \in V(H) - \{s_1, \cdots, s_k\}
\end{align*}
\end{boxedminipage}
\end{center}

\noindent
We map the solution $\bx$ to \CR to a solution $\bd$ to
\textsc{Distance-LP} as follows. For each vertex
$v \in V$, we set $d_v = 0$ (recall that all the vertices in $V$ have
infinite weight). Let $z$ be a vertex of $H$ that corresponds to the
hyperedge $e$ of $G$. Let $u$ be the representative of $e$, and let
	$$d_z = \sum_{i = 1}^k \left(x(u, i) - \min_{v \in e} x(v,
	i)\right) = 1 - \sum_{i = 1}^k \min_{v \in e} x(v, i)$$
Now consider a path $p$ of $H$ between two terminals $s_a$ and $s_b$.
Let $p = v_0 - z_1 - v_1 - z_2 - \cdots - z_{\ell} - v_{\ell}$, where
$z_j$ is in $V(H) - V$, $v_j$ is in $V$, $v_0 = s_a$, and $v_{\ell} =
s_b$. Let $e_j$ be the hyperedge corresponding to $z_j$; node that
$e_j$ contains $v_{j - 1}$ and $v_j$.  For each vertex $v \in V$, let
$\overline{x}(v) = (x(v, 1), \cdots, x(v, k))$ denote the point on the
$k$-dimensional simplex to which $v$ is mapped by the solution $\bx$,
and let $||\cdot||_1$ denote the $\ell_1$ norm of a vector.

\begin{eqnarray*}
	\sum_{j = 1}^{\ell} d_{z_j} &=& \sum_{j = 1}^{\ell} \left( 1 -
	\sum_{i = 1}^k \min_{v \in e_j} x(v, i)\right)\\
	&\geq& \sum_{j = 1}^{\ell} \left( 1 - \sum_{i = 1}^k \min\{x(v_{j - 1}, i),
	x(v_j, i)\}\right) \qquad \mbox{($e_j$ contains
	$v_{j - 1}$ and $v_j$)}\\
	&=& \sum_{j = 1}^{\ell} \sum_{i = 1}^k \left(x(v_{j - 1}, i) -
	\min\{x(v_{j - 1}, i), x(v_j, i)\}\right)\\
	&=& \sum_{j = 1}^{\ell} \sum_{i = 1}^k \max\{0, x(v_{j - 1}, i)\ -
	x(v_j, i)\}\\
	&=& {1 \over 2} \sum_{j = 1}^{\ell} \sum_{i = 1}^k |x(v_{j - 1},
	i) - x(v_j, i)|\\
	&=& {1 \over 2} ||\overline{x}(v_{j - 1}) -
	\overline{x}(v_j)||_1\\
	&\geq& {1 \over 2} ||\overline{x}(s_a)  - \overline{x}(s_b)||_1 = 1
\end{eqnarray*}
Therefore $\bd$ is a feasible solution to \textsc{Distance-LP}. Garg,
Vazirany, and Yannakakis \cite{GargVY04} showed that the integrality
gap of \textsc{Distance-LP} is at most $2(1 - 1/k)$. Therefore the
integrality gap of \CR is at most $2(1 - 1/k)$ as well. The above
argument
also establishes that \CR is at least as strong a relaxation as
\textsc{Distance-LP} for \nodeMC. Easy examples show that it is strictly
stronger.
\end{document}